\documentclass[onefignum,onetabnum]{siamart190516}

\usepackage[title]{appendix}

\usepackage{lineno,hyperref}
\usepackage{graphicx}
\usepackage{subcaption}
\usepackage{tikz}
\usetikzlibrary{arrows,calc}
\usetikzlibrary{patterns}
\tikzset{
	>=stealth',
	help lines/.style={dashed, thick},
	axis/.style={<->},
	important line/.style={thick},
	connection/.style={dotted},
}

\usepackage{algorithm}
\usepackage{algpseudocode}
\usepackage{multirow}

\def\scl{T}

\def\half{\mbox{$\frac{1}{2}$}}
\def\x{\mathbf{x}}

\def\e{\mathbf{\e}}

\def\0{\mathbf{0}}

\def\eps{\epsilon}

\newcommand\numbereq{\addtocounter{equation}{1}\tag{\theequation}}

\title{A Fully Polynomial Time Approximation Scheme for the Replenishment Storage Problem}

\author{Dorit S. Hochbaum\thanks{Department of IEOR, Etcheverry Hall, Berkeley, CA, supported in part by NSF award No. CMMI-1760102.
		(\email{hochbaum@ieor.berkeley.edu}).}
	\and Xu Rao\thanks{Department of IEOR, Etcheverry Hall, Berkeley, CA.
		(\email{xrao@berkeley.edu}).}}

\headers{}{D. S. Hochbaum and X. Rao}

\ifpdf
\hypersetup{
	pdftitle={},
	pdfauthor={D. S. Hochbaum and X. Rao}
}
\fi

\begin{document}
	
	\maketitle
	
	\begin{abstract}
		The Replenishment Storage problem (RSP) is to minimize the storage capacity requirement for a deterministic demand, multi-item inventory system where each item has a given reorder size and cycle length. The reorders can only take place at integer time units within the cycle. This problem was shown to be weakly NP-hard for constant joint cycle length (the least common multiple of the lengths of all individual cycles).
		When all items have the same constant cycle length, there exists a Fully Polynomial Time Approximation Scheme (FPTAS), but no FPTAS has been known for the case when the individual cycles are different. Here we devise the first known FPTAS for the RSP with different individual cycles and constant joint cycle length.
	\end{abstract}
	

\section{Introduction}

The Replenishment Storage problem (RSP) arises in planning a periodic replenishment schedule of multiple items so as to minimize the storage capacity required.
The input to the RSP consists of in a multi-item inventory system where each item has deterministic demand, a given reorder size and its own cycle length determined by its Economic Order Quantity.
Here the reorders can only take place at an integer time unit within the cycle. The problem is to determine the timing of the first replenishment of each item within its cycle so that the maximum inventory level of all items over time is minimized.

An instance of RSP consists of $n$ items. Each item $i$ is associated with an integer individual cycle length $k_i$, and an integer reorder size $s_i$. Here $s_i$ is expressed in terms of the storage amount required for the reorder quantity.  The {\em joint cycle length} of the $n$ items is the least common multiple (lcm) of the lengths $k_i$, $i=1,\ldots,n$.   We let $k=\mbox{lcm}(k_1,\ldots,k_n)$. By the cyclical nature of the problem, the total inventory levels repeat periodically every $k$ units of time for any reorder schedule.  If all items have the same cycle length, $k$, the problem is said to be {\em single-cycle}, otherwise it is said to be {\em multi-cycle}.

The RSP is an NP-hard problem \cite{Hall,HR19}, so there is no polynomial time optimization algorithm unless $P=NP$. But a polynomial time approximation scheme may exist for the problem.
An approximation scheme is a family of $(1+\eps)$-approximation algorithms for every $\eps >0$. If the running time is polynomial in the problem size for every fixed $\eps$, then this scheme is a Polynomial Time Approximation Scheme (PTAS); furthermore, if the running time is polynomial in both the problem size and $1/\eps$, then it is a Fully Polynomial Time Approximation Scheme (FPTAS).
Hochbaum and Rao \cite{HR19} gave a Fully Polynomial Time Approximation Scheme (FPTAS) for the single-cycle RSP when $k$ is a constant \cite{HR19}. For the multi-cycle case however no FPTAS has been known to date. Here, we establish for the first time an FPTAS for the multi-cycle RSP when the joint cycle length, $k$, is constant. We also observe here that the FPTAS of Hochbaum and Rao for the single cycle RSP is fixed-parameter tractable (FPT) and is in fact linear for a constant length of the single cycle.


\subsection{Related Literature}

The single-cycle RSP was shown by Hall \cite{Hall} to be NP-hard, even when the joint cycle length $k=2$. Since the single-cycle RSP is a special case of the multi-cycle RSP, it implies that the multi-cycle RSP is also NP-hard, even when $k$ is small.
Hochbaum and Rao \cite{HR19} investigated the complexity status of the single-cycle and the multi-cycle RSPs and showed that the problems are strongly NP-hard when $k$ is not a constant, but weakly NP-hard when $k$ is a constant. They further provided in \cite{HR19} a pseudo-polynomial optimization algorithm for the two problems.

These complexity results imply that there is no polynomial time algorithm for single-cycle and the multi-cycle RSPs even when $k$ is a constant, unless $P=NP$.
To that end several approximation results have been delivered for the single-cycle RSP. Hall \cite{Hall} provided a linear time approximation algorithm for the single-cycle RSP, with an approximation factor of $\left(1+\frac{2}{k}\right)$, even for non-constant $k$. Hochbaum and Rao \cite{HR19} devised for the single-cycle RSP with constant $k$ an FPTAS, and for the single-cycle RSP with non-constant $k$ a Polynomial Time Approximation Scheme (PTAS). The complexity of the FPTAS for a $(1+\eps)$-approximation algorithm is $O(\frac{n}{\eps^{2k}})$ \footnote{There was a mistake in the proof of Theorem 6 in \cite{HR19}, but this can be addressed by replacing the original scaling factor $\eps^2 D$ by $\eps^2 k D$ in the FPTAS. The running time is only affected by a constant factor $k^k$ so the approximation scheme is still an FPTAS. What's more, we observe here that this running time is fixed-parameter tractable for parameter $k$.}, and the complexity of the PTAS for non-constant cycle length is $O((\frac{2}{\eps})!\cdot \frac{n}{\eps^{2k+2}})$.

For the multi-cycle RSP with only two items, Murthy et. al. \cite{MBR} provided an optimal closed-form replenishment solution, meaning that it is solved in constant time.
Studies of algorithmic results for the multi-cycle RSP with more than two items have been focused on the development of heuristics.  These include genetic algorithms \cite{MCK, YC}), a smoothing procedure utilizing a Boltzmann function \cite{YCL}, local-search procedures \cite{CH}, a simulated-annealing algorithm \cite{Boctor} and a hybrid heuristic \cite{Boctor,RU}. No algorithm with guaranteed approximation bound has been known for the multi-cycle RSP.

\subsection{Contributions}
A weakly NP-hard problem can have an FPTAS and it was shown in \cite{HR19} that for constant $k$ the RSP problem is weakly NP-hard.   For constant parameter $k$, Hochbaum and Rao \cite{HR19} devised for the single-cycle RSP an FPTAS, which we observe here is fixed-parameter tractable (FPT).   We devise here an FPTAS for the multi-cycle RSP with constant joint cycle length for the first time. Unlike the case of the single-cycle (in \cite{HR19}), the running time of this FPTAS for the multi-cycle RSP is not fixed-parameter tractable for parameter $k$.

A summary of the complexity results for RSP that includes our contributions here is given in Table \ref{table:NP}.

\begin{table}[h!]
	\centering
	\small
	\caption{Summary of complexity and algorithmic results for the RSP.}
	\begin{tabular}{ |c|c|c| }
		\hline
		Problem & non-constant joint cycle & constant joint cycle  \\
		\hline
		\multirow{3}{*}{single-cycle}
		& strongly NP-hard &  weakly NP-hard   \\
		& $(1+2/k)$-approximation \cite{Hall},  &   $(1+2/k)$-approximation \cite{Hall}, \\
		& PTAS \cite{HR19} &  pseudo-poly algorithm  \cite{HR19} \& FPTAS (here FPT in $k$)\cite{HR19} \\
		\hline
		\multirow{2}{*}{multi-cycle}
		 &   strongly NP-hard&  weakly NP-hard  \\
		 &   - &  pseudo-poly algorithm \cite{HR19} \& FPTAS (here)  \\
		\hline
	\end{tabular}
	\label{table:NP}
\end{table}

\subsection{Paper Overview}
The next section, Section \ref{sec:Preliminaries}, introduces the notation, an integer programming formulation as well as a pseudo-polynomial algorithm for the RSP which is relevant to the approximation scheme.  In Section \ref{sec:FPTAS-multi} we describe the new fully polynomial-time approximation scheme (FPTAS) for the multi-cycle RSP for constant joint cycle length $k$.

\section{Preliminaries}\label{sec:Preliminaries}

Our approximation scheme utilizes a dynamic programming algorithm for the RSP derived by Hochbaum and Rao \cite{HR19}. That dynamic programming algorithm uses an integer programming (IP) formulation of the RSP that was introduced in \cite{HR19}.
Since this algorithm and IP formulation are crucial for our FPTAS, we sketch them here.

We first present necessary notation. For an instance of RSP, the demand rates and inventory levels are given in terms of the respective reorder size: for item $i$, the demand per unit of time is $\frac{s_i}{k_i}$, and its inventory levels at each replenishment cycle of $k_i$ time units starting at time $T$, $(T+0,T+1,\ldots ,T+ k_i-1)$, are $(s_i, \frac{k_i-1}{k_i} s_i, \frac{k_i-2}{k_i} s_i,\ldots,\frac{1}{k_i}s_i)$.
Recall that since $k=\mbox{lcm}(k_1,\ldots,k_n)$, the inventory levels are periodic within a cycle of $k$ time units (repeat every $k$ time units). It is therefore sufficient to determine the peak storage requirement by examining a time interval of length $k$. This is because each item must be reordered at least once in such interval, and the peak storage always coincides with the reorder timing of an item. (Note that inventory level at time $0$ is the same as inventory level at time $k$.)

The decision variables in the integer programming formulations are the assignments of time periods within the $k$-unit time frame to the orders of all items.  This assignment of timing is given as an $n\times k$ binary matrix $\x$ where
\[ x_{ij}=\left\{ \begin{array}{cc}
1& \mbox{       if item } i \mbox{ is ordered at time } $j$, \\
0&  \mbox{otherwise.}
\end{array}
\right.  \]
\begin{definition}\label{defvad}
	A $n\times k$ binary matrix $\x$
	is said to be a {\em valid assignment} for a given instance if and only if each item $i$ is replenished exactly once every $k_i$ time units. That is,
	\begin{center}
		$\sum_{j=1}^{k_i} x_{ij} =1\ \ \quad   i =1,\ldots ,n,\quad \mbox{and} \quad x_{ij}=x_{i,(j-k_i)} \quad i=1,\ldots,n,\ \ j=k_i+1,\ldots,k.  $
	\end{center}
\end{definition}

The following lists the notation for demand rates, inventory levels, the total sum of reorder sizes at an integer time and the optimal peak storage:
\\ $d_i=\frac{s_i}{k_i}$: demand rate of item $i$ for $i=1,...,n$.
\\ $D=\sum _{i=1}^n d_i= \sum _{i=1}^n \frac{s_i}{k_i}$: total demand (aggregate stock depletion) per unit of time.
\\ $V_\ell(\x)$: the inventory level at time $\ell$ according to assignment $\x$ for $\ell=1,...,k$.
\\ $V(\x)=\max_{\ell\in \{1,..,k\}} V_\ell(\x)$: the maximum inventory level (peak storage) of a cycle.
\\ $Q_j(\x)=\sum_{i=1}^{n} s_ix_{ij}$: the total sum of reorder sizes at time $j$ for $j=1,...,k$.
\\ $V^* = \min _{\x {\rm {\ valid}}} V(\x)$: the optimal peak inventory level.
\\ Let the following quantity, which is a constant, be denoted by $C$: $C=\sum_{i=1}^n \half (1+\frac{1}{k_i})ks_i+\frac{(1+k)k}{2}D$.  This quantity is used in the IP formulation and the FPTAS.

\subsection{The Integer Programming Formulation of the RSP}\label{sec:IP}
The IP formulation of Hochbaum and Rao \cite{HR19} is based on three lemmas derived in their paper, which are included for the sake of completion. Lemma \ref{lem:permutation} shows that valid assignments can be restriced to those attaining peak inventory level at time $k$ without changing the optimal solution of the RSP. Lemma \ref{lem:inventory level} establishes the relation between the inventory levels $V_\ell(\x)$ for $\ell=1,...,k$ and the total amount ordered at time $j$, $Q_j(\x)$ for $j=1,...,k$.
Let $z(\x)$ be the following function of a valid assignment $\x$:
\[z(\x)=\sum _{j=1}^k (k-j+1) Q_j(\x)= \sum _{j=1}^k (k-j+1)\sum_{i=1}^n s_ix_{ij}.\]
Lemma \ref{lem:objectivez} shows that minimizing the inventory level of time $k$, $V_k(\x)$, is equivalent to maximizing $z(\x)$.

\begin{lemma}[\cite{HR19}]\label{lem:permutation}
	For any valid assignment $\x$ there is a shift-permutation of $1,\ldots,k$, denoted by $\pi(1),\ldots,\pi(k)$, such that the valid assignment $\x'$ with $x'_{ij}=x_{i\pi(j)}$, attains peak inventory level at time $k$, and this new peak inventory level equals the peak inventory level of assignment $\x$. That is, $V_k(\x')=V(\x')=V(\x)$.
\end{lemma}

\begin{lemma}[\cite{HR19}]\label{lem:inventory level}
	For any valid assignment $\x$,	
	\begin{equation}\label{Vell}
	V_\ell(\x)= V_k(\x)-\ell D+\sum _{j=1}^{\ell} Q_j(\x), \quad  \ell =1,..,k
	\end{equation}
\end{lemma}

\begin{lemma}[\cite{HR19}]\label{lem:objectivez}
	For any valid assignment $\x$,	
	$kV_k(\x)+z(\x)=C$ where $C$ is a constant defined as $\sum_{i=1}^n \half (1+\frac{1}{k_i})ks_i+\frac{(1+k)k}{2}D$.
\end{lemma}

Restricting valid assignments to those attaining peak inventory level at time $k$ does not change the optimal solution of the RSP. So the RSP can be formulated as minimizing the inventory level at time $k$ such that the schedule is a valid assignment that attains peak inventory level at time $k$, which can be written as $V_\ell(\x) \leq V_k(\x)$ for $\ell =1,...,k$. These inequalities, according to Lemma \ref{lem:inventory level}, are equivalent to,
\begin{equation}\label{cascading}
\sum _{j=1}^{\ell} Q_j(\x)\leq \ell D \mbox{ for } \ell =1,..,k
\end{equation}
This set of inequalities (\ref{cascading}) are referred to as the {\em cascading constraints}. These constraints enforce the peak storage to occur at time $k$.
From Lemma \ref{lem:objectivez}, we know that minimizing the inventory level of time $k$, $V_k(\x)$, is equivalent to maximizing $z(\x)$ as the sum of $kV_k(\x)$ and $z(\x)$ is a constant $C$ defined earlier.
Therefore, the below integer programming formulation (RSP) derived by \cite{HR19} solves the RSP.  For presentation simplicity we use $Q_j(\x)=\sum_{i=1}^n s_i x_{ij}$:

\[
\hspace{.4in}\begin{array}{ll}
\mbox{(RSP)~~~~}  \max \ &
z(\x)=\sum _{j=1}^k (k-j+1)Q_j(\x)\\
\mbox{subject to }\ & \sum _{j=1}^{\ell} Q_j(\x)\leq \ell D
\quad \ell =1,..,k \\
& \sum_{j=1}^{k_i} x_{ij} =1\ \
\quad   i =1,\ldots ,n \\
& x_{ij}=x_{i,(j-k_i)} \quad i=1,\ldots,n,\ \ j=k_i+1,\ldots,k \\
& x_{ij} {\mbox { binary for }} i=1,...,n,\ \ j=1,..,k_i.
\end{array}
\]

\subsection{The dynamic programming algorithm for the RSP }\label{sec:DP}

We present here the dynamic programming algorithm of Hochbaum and Rao \cite{HR19}, which is associated with the IP formulation (RSP).
For $h$ an integer such that $0\leq h\leq n$,  let $\x^h$ denote the assignment of reorders for the first $h$ items.  Let the function $ f_h(q_1,q_2 ,...,q_k)$ be the maximum of $z(\x^h)$ with the cumulative reorder sizes at time $\ell$ being restricted to less than or equal to $q_\ell$ for $\ell=1,...,k$. Here, $(q_1,\ldots ,q_k)$ is an integer array with $q_\ell \in [0, \ell D]$.  Formally,

\[
\hspace{.4in}\begin{array}{lll}
f_h(q_1,q_2 ,...,q_k)= &\max \ &
\sum _{j=1}^k (k-j+1)Q_j(\x^h)\\
&\mbox{subject to }\ & \sum _{j=1}^{\ell} Q_j(\x^h) \leq q_\ell
\quad \ell =1,..,k \\
&& \sum_{j=1}^{k_i} x_{ij} =1\ \
\quad   i =1,\ldots ,h \\
&& x_{ij}=x_{i(j-k_i)} \quad i=1,\ldots,h,\ \ j=k_i+1,\ldots,k \\
&& x_{ij} {\mbox { binary for }} i=1,...,h,\ \ j=1,..,k_i,
\end{array}
\]
where $Q_j(\x^h)=\sum_{i=1}^h s_i x_{ij}$.
The function $f_h(q_1,q_2 ,...,q_k)$ is set to $0\infty$ if the above integer programming problem is infeasible. The optimal solution being sought is $ f_n(D,2D,...,kD)$.

The values of the function $ f_h(q_1,q_2 ,...,q_k)$ are evaluated for every  $0\leq h\leq n$ and any integer array $(q_1,\ldots ,q_k)$, where  $q_j \in [0, jD]$, with a dynamic programming recursion. The boundary conditions are $f_0 (q_1,q_2,\ldots,q_k)=0$ for any $(q_1,q_2,\ldots,q_k)$.
The recursive derivation of $f_h(q_1,q_2,\ldots,q_k)$ from $f_{h-1}(\cdot)$ requires to determine the timing to replenish item $h$ within the first $k_h$ time units so as to maximize the objective $\sum _{j=1}^k (k-j+1)Q_j(\x^h) $.
The recursive equation, using the notation $q'_\ell(\tau)=q_\ell-\lfloor{\frac{\ell-\tau+k_h}{k_h}}\rfloor s_h$, is:
\[
\small
f_h(q_1,q_2 ,...,q_k)= \left\{ \begin{array}{ll}
\max_{\tau=1,...,k_h} \{\left(\frac{k+k_h}{2}+1-\tau \right)\frac{k}{k_h}s_h+f_{h-1}(q'_1(\tau),...,q'_k(\tau))\} , & \mbox{ if $q'_\ell(\tau)\geq 0$ for all $\ell$ } \\
-\infty & \mbox{ otherwise.}
\end{array}
\right.
\]

All function values are evaluated recursively for $h=1,...,n$ and for all integer values of $(q_1,\ldots ,q_k)$, where each $q_j \in [0, jD]$ and $q_j$ integer.  Each function evaluation is associated with a choice of $\tau(h)$, which is the timing of the replenishment of item $h$ within the $k_h$ cycle.  The optimal objective value is then $ f_n(D,2D,...,kD)$.
To recover the optimal valid assignment we record the choices of the replenishment timings within the $k$ cycle, for each function value evaluation.

The running time of this algorithm is $O(nD^k)$ for constant $k$ \cite{HR19}, which is pseudo-polynomial as it depends on the value $D$.

\section{A fully polynomial-time approximation scheme for the RSP with constant joint cycle length}\label{sec:FPTAS-multi}

As the RSP is strongly NP-hard when the joint cycle length is not a constant, there is no fully polynomial-time approximation scheme assuming that $P\neq NP$. However, when the joint cycle length $k$ is constant, it is possible to obtain a fully polynomial-time approximation scheme for this problem. Hochbaum and Rao \cite{HR19} showed an FPTAS for the single-cycle RSP but no FPTAS has been known for the multi-cycle case when $k$ is constant.
In this section, we establish the first known FPTAS for the multi-cycle RSP for constant joint cycle length.

Here we derive a family of $(1+\eps')$-approximation algorithms for the multi-cycle RSP for every $\eps' >0$.
The $(1+\eps')$-approximation algorithm works by applying the dynamic programming algorithm in Section \ref{sec:DP} with scaled reorder sizes with some scaling factor $\scl$. We show in this section that the output of the dynamic programming algorithm using the scaled sizes is within a factor of $1+\eps'$ of the optimal solution.  The run time of this approximation algorithm is polynomial in $n$ and $\frac{1}{\eps'}$, and hence this family of algorithms is a fully polynomial approximation scheme.

%

The approximation algorithm solves a modified RSP, (scaled-RSP), in which the order sizes are scaled by a factor $\scl$.
The scaled problem is solvable using the dynamic programming procedure of Section \ref{sec:DP} and the solution of it is a valid assignment that has objective function value close to the optimal value of (RSP).

\subsection{The scaling of (RSP), (scaled-RSP)}

For any $\eps'>0$, we let $\eps= \eps'/2$ and we scale the reorder sizes by the factor $\scl =  \frac{\eps D}{kn}$ as follows.
Let $s'_i= \lfloor{\frac{s_i}{\scl}}\rfloor$ be the scaled sizes of items $i=1,...,n$ and $D'=\frac{D}{\scl}$ be the scaled demand.
Let $Q'_j(\x)$ and $z'(\x)$ denote the ``scaled" replenishment sizes at time $j$ and the objective function for the scaled sizes $s'_i$:
$Q'_j(\x)=\sum_{i=1}^n s'_i x_{ij}, \ \ j =1,..,k$;
$z'(\x)=\sum _{j=1}^k (k-j+1) Q'_j(\x)$.
\vspace{1em}

The scaled problem (scaled-RSP) is formulated as follows:

\[
\hspace{.4in}\begin{array}{ll}
\mbox{(scaled-RSP)~~~~}  \max \ &
z'(\x)=\sum _{j=1}^k (k-j+1) Q'_j(\x)\\
\mbox{subject to }\ &   \sum _{j=1}^{\ell}Q'_j(\x)\leq \ell D'
\quad \ell =1,..,k \\
& \sum_{j=1}^{k_i} x_{ij} =1\ \
\quad   i =1,\ldots ,n \\
& x_{ij}=x_{i,(j-k_i)} \quad i=1,\ldots,n,\ \ j=k_i+1,\ldots,k \\
& x_{ij} {\mbox { binary for }} i=1,...,n,\ \ j=1,..,k_i.
\end{array}
\]

The optimal solution for (scaled-RSP) is found by applying the dynamic programming procedure in Section \ref{sec:DP} with scaled sizes $D'$ and $s'_1,\ldots,s'_{n}$.

The running time of finding the optimal solution for  (scaled-RSP) with the dynamic programming procedure, is $O(n D'^k )=O(\frac{n^{k+1}}{\eps^{k}})$.

Next we define the ($\eps$-relaxed RSP) and then prove that any feasible solution for (scaled-RSP), including $\hat{\x}$, is feasible for ($\eps$-relaxed RSP).

\subsection{The $\eps$-relaxed RSP}\label{sec:relaxed}

The ($\eps$-relaxed RSP) formulation allows the cascading constraints to be violated by up to $ \eps D$ as follows:

\[
\hspace{.4in}\begin{array}{ll}
(\eps\mbox{-relaxed RSP)~~~~}  \max \ &
z(\x)=\sum _{j=1}^k (k-j+1)Q_j(\x)\\
\mbox{subject to }\ & \sum _{j=1}^{\ell} Q_j(\x) \leq \ell D+ \eps D
\quad \ell =1,..,k \\
& \sum_{j=1}^{k_i} x_{ij} =1\ \
\quad   i =1,\ldots ,n \\
& x_{ij}=x_{i,(j-k_i)} \quad i=1,\ldots,n,\ \ j=k_i+1,\ldots,k \\
& x_{ij} {\mbox { binary for }} i=1,...,n,\ \ j=1,..,k_i.
\end{array}
\]
We refer to the constraints $\sum _{j=1}^{\ell} Q_j(\x) \leq \ell D+ \eps D$ as the {\em $\eps$-relaxed cascading constraints}. We next show that the effect of the  $\eps$-relaxed cascading constraints on the optimal solution is at most $ \eps D$.

\begin{lemma}\label{lem:relaxedV}
 The peak inventory level of any feasible solution $\x$ to ($\eps$-relaxed RSP) is at most $V_k(\x)+ \eps D$.
\end{lemma}

\begin{proof}
Any feasible solution $\x$ for ($\eps$-relaxed RSP) is a valid assignment, so Lemma \ref{lem:inventory level} applies. That is, $V_\ell({\x})= V_k({\x})+\left(\sum _{j=1}^{\ell} Q_j({\x})-\ell D\right)$ for $\ell=1,...,k$. The $\eps$-relaxed cascading constraints state that $\sum _{j=1}^{\ell} Q_j({\x}) -\ell D \leq \eps D$	for all $\ell$. So when $\x$ is a feasible solution of ($\eps$-relaxed RSP), $V_\ell({\x})\leq V_k({\x})+ \eps D$ for all $\ell$,
and hence, $V({\x})=\max_{\ell} V_{\ell}({\x})\leq V_k({\x})+ \eps D$.
	
\end{proof}

The next lemma proves that any feasible solution for (scaled-RSP), including $\hat{\x}$, is feasible for ($\eps$-relaxed RSP).

\begin{lemma} \label{lem:fea}
	Any assignment $\x$ that is feasible for (scaled-RSP) is feasible for ($\eps$-relaxed RSP).
\end{lemma}

\begin{proof}
In both problems $\x$ is required to be a valid assignment.
It remains to show that $\x$ satisfies the $\eps$-relaxed cascading constraints, that is,
$\sum _{j=1}^{\ell} Q_j(\x) \leq \ell D+\eps D$  for $\ell =1,..,k$.

	By definition, $s'_i= \lfloor{\frac{s_i}{\scl}}\rfloor$. So $s_i<\scl (s'_i +1)$ and thus,
	\begin{equation}\label{l91}
	\sum _{j=1}^{\ell} Q_j(\x)=\sum_{j=1}^{\ell} \sum_{i=1}^{n} s_i \x_{ij} \leq  \sum_{j=1}^{\ell} \sum_{i=1}^{n} \scl(s'_i+1) \x_{ij}=\scl\left( \sum_{j=1}^{\ell} \sum_{i=1}^{n} s'_i \x_{ij}+  \sum_{j=1}^{\ell} \sum_{i=1}^{n}\x_{ij} \right).
	\end{equation}

	Since $\x$ is feasible for (scaled-RSP) and $D'=\frac{D}{\scl}$,
	\begin{equation}\label{l92}
	\sum_{j=1}^{\ell} \sum_{i=1}^{n} s'_i \x_{ij} = 	\sum _{j=1}^{\ell} Q'_j(\x) \leq \ell D'=\frac{\ell D}{\scl}.
	\end{equation}	
	
	For $\ell=1, ..., k$,
		\begin{equation}\label{l93}
		\sum_{j=1}^{\ell} \sum_{i=1}^{n}\x_{ij}\leq 	\sum_{j=1}^{k} \sum_{i=1}^{n}\x_{ij} \leq nk
		\end{equation}

	Hence from inequalities (\ref{l91}), (\ref{l92}) and (\ref{l93}),
	\[	\sum _{j=1}^{\ell} Q_j(\x)\leq \scl  \left(\frac{\ell D}{\scl}+ nk \right) =\ell D+\eps D.\]
	
\end{proof}

Using the relationship between reorder sizes $s_i$ and the scaled sizes $s'_i$, we show that for any feasible solution of (scaled-RSP),  $\x$, the objective with original sizes $z(\x)= \sum _{j=1}^k (k-j+1) Q_j(\x)$ is closely approximated by the objective with scaled sizes $z'(\x)= \sum _{j=1}^k (k-j+1) Q'_j(\x)$, corrected for the scaling factor $\scl$:

\subsection{The approximation property of the solution to (scaled-RSP)}

\begin{lemma}\label{lem:zz'}
	For any assignment of items $\x$ feasible for (scaled-RSP), the values of the objective function with original and scaled sizes, $z(\x)$ and $z'(\x)$ respectively, satisfy,
	
	\[ \scl z'(\x)\ \leq\  z(\x) \ \leq\   \scl z'(\x)+ \eps kD. \]
\end{lemma}

\begin{proof}
	
	Recall that $s'_i= \lfloor{\frac{s_i}{\scl}}\rfloor$, so $\scl s'_i\leq s_i<\scl  (s'_i +1)$. We derive the lower bound on $z(\x)$ as follows:
	
	\begin{align*}
	z({\x}) &= \sum _{j=1}^k (k-j+1) Q_j({\x})\\
	&= \sum _{j=1}^k (k-j+1) \sum_{i=1}^{n} s_i{x}_{ij}\\
	& \geq \scl \cdot  \sum _{j=1}^k (k-j+1)\sum_{i=1}^{n}  s'_i {x}_{ij}\\
	& = \scl \cdot \sum _{j=1}^k (k-j+1)) Q'_j({\x})\\
	&= \scl  z'({\x}) . \numbereq \label{l1}
	\end{align*}
	
	The upper bound on $z(\x)$ can be derived as follows:
	
	\begin{align*}
	z({\x}) &= \sum _{j=1}^k (k-j+1) Q_j({\x})\\
	&= \sum _{j=1}^k (k-j+1) \sum_{i=1}^{n} s_i{x}_{ij}\\
	& \leq \scl \cdot  \sum _{j=1}^k (k-j+1)\sum_{i=1}^{n} ( s'_i+1) {x}_{ij}\\
	& = \scl \cdot \left[ \sum _{j=1}^k (k-j+1) Q'_j({\x}) +\sum _{j=1}^k (k-j+1) \sum_{i=1}^{n}{x}_{ij} \right] \\
	&\leq \scl  z'({\x}) + \scl k^2n \\
	&= \scl z'({\x}) + \eps kD . \numbereq \label{u1}
\end{align*}

\end{proof}

Lemma \ref{lem:zz'} leads to the following lower bound on $z(\hat{\x})$ for $\hat{\x}$ being an optimal solution of (scaled-RSP):

\begin{theorem}\label{lem:superopt20}
	For any feasible solution $\x$ of (RSP), $z(\hat{\x})\geq z(\x)-\eps kD$.
\end{theorem}

\begin{proof}

	By Lemma \ref{lem:zz'}, we know
	$ z(\hat{\x})\geq \scl z'(\hat{\x})$.
Since any feasible solution of (RSP), $\x$, is also feasible for (scaled-RSP), we use the upper bound of $z(\x)$  from Lemma \ref{lem:zz'} to get:
	\[\scl  z'(\x) \geq  z(\x)-\eps kD.\]
Because $\hat{\x}$ is optimal for (scaled-RSP), it follows that $z'(\hat{\x})\geq z'(\x)$. Combining the three inequalities, we get
\[
z(\hat{\x}) \geq \scl  z'(\hat{\x}) \geq \scl   z'(\x) \geq z(\x)-\eps kD.
\]
\end{proof}

Consequently, the optimal solution $\hat{\x}$ for (scaled-RSP) attains a objective value $z(\hat{\x})$ that is at least as much as the optimal objective of (RSP) minus $\eps kD$.

\subsection{The $(1+\eps')$-approximation bound}

From the discussion above, we know that the optimal solution for (scaled-RSP) $\hat{\x}$ is a valid assignment whose inventory levels at time $k$ approximates that maximum inventory level, and the value $z(\x)$ approximates the optimal objective value of (RSP). We will prove here that $\hat{\x}$ is an $(1+\eps ')$-approximation solution for $\eps' =2\eps$ and any $\eps>0$.

%

\begin{theorem}
	The optimal solution $\hat{\x}$ for (scaled-RSP) is a $(1+\eps')$-approximation solution for the RSP.
\end{theorem}
\begin{proof}
	Assignment $\hat{\x}$ is valid as it is feasible for (scaled-RSP). So we just need to prove the approximation factor for the peak inventory level.

	Let $\x^*$ be an optimal solution of (RSP), and $V^*$ the corresponding peak inventory level.
	
	As stated in Theorem \ref{lem:superopt20}, $z(\hat{\x}) \geq z(\x)-\eps kD $ for any $\x$ that is feasible of (RSP), including $\x^*$.   From Lemma \ref{lem:objectivez}, the inventory levels at time $k$ for $\hat{\x}$ and $\x^*$ are $ V_k(\hat{\x})=\frac{C}{k}-\frac{z(\hat{\x})}{k}$ and $ V_k(\x^*)=\frac{C}{k}-\frac{z(\x^*)}{k}$ respectively. Therefore,	 
	\[ V_k(\hat{\x})=\frac{C}{k}-\frac{z(\hat{\x})}{k} \leq \frac{C}{k}-\frac{z(\x^*)}{k}+\frac{\eps kD}{k} =V_k(\x^*)  +\eps D.\]
	
	From Lemma \ref{lem:relaxedV} it follows that the peak inventory level for $\hat{\x}$ satisfies $ V(\hat{\x})\leq V_k(\hat{\x})+\eps D$.  Since $\x^*$ is a solution of (RSP), the peak inventory level for $\x^*$ is $V^*=V_k(\x^*)$. Hence,
	\[ V(\hat{\x})\leq V_k(\hat{\x})+\eps D \leq V^*  + 2\eps D . \]
	
	That is, for the optimum peak storage of (RSP), $V^*$, and for the optimal solution of (scaled-RSP) $\hat{\x}$, the ratio $V(\hat{\x})/V^*$ is at most $1+ 2\eps D/V^*$.
	Observe that $V^*$ must be at least the per unit time demand $D$, it follows that
	$2\eps D/V^*  \leq  {2\eps}$.

	Therefore, the ratio $V(\hat{\x})/V^*$ is at most $1+2\eps=1+\eps'$. Hence, $\hat{\x}$ is a $(1+\eps')$-approximate solution to the RSP.
\end{proof}

The complexity of this approximation procedure is $O(\frac{n^{(k+1)}}{\eps^{k}})$ for constant $k$. Noted that $\frac{1}{\eps}=O(\frac{1}{\eps'})$. Therefore the complexity of the RSP $(1+\eps')$-approximation algorithm is $O(\frac{n^{(k+1)}}{\eps'^{k}})$, which is polynomial in $n$ and $\frac{1}{\eps'}$ for constant $k$. And a family of $(1+\eps ')$-approximation algorithms with complexity that is polynomial in $n$ and $\frac{1}{\eps'}$ is called a Fully Polynomial Time Approximation Scheme.

\section{Concluding Remarks}\label{sec:conclude}

Both the single-cycle and the multi-cycle RSPs are weakly NP-hard but an FPTAS was known only for the single cycle RSP, in \cite{HR19}.  Here we devise an FPTAS for the {\em multi-cycle} RSP with constant joint cycle length. The running time of our FPTAS here is not fixed-parameter tractable as compared to the running time of the FPTAS for the single-cycle case. We leave the existence of a fixed-parameter tractable FPTAS for the multi-cycle RSP as an open question. The question of whether there exists a PTAS for the multi-cycle RSP when the joint cycle length is not constant remains open as well.

%

\bibliographystyle{siamplain}
\bibliography{mybibfile}

\begin{thebibliography}{1}

\bibitem{Boctor}
{\sc F.~F. Boctor}, {\em Offsetting inventory replenishment cycles to minimize
  storage space}, European Journal of Operational Research, 203 (2010),
  pp.~321--325.

\bibitem{CH}
{\sc E.~Croot and K.~Huang}, {\em A class of random algorithms for inventory
  cycle offsetting}, International Journal of Operational Research, 18 (2013),
  pp.~201--217.

\bibitem{Hall}
{\sc N.~G. Hall}, {\em A comparison of inventory replenishment heuristics for
  minimizing maximum storage}, American Journal of Mathematical and Management
  Sciences, 18 (1998), pp.~245--258.

\bibitem{HR19}
{\sc D.~S. Hochbaum and X.~Rao}, {\em The replenishment schedule to minimize
  peak storage problem: The gap between the continuous and discrete versions of
  the problem}, Operations Research, 67 (2019), pp.~1345--1361.

\bibitem{MCK}
{\sc I.~K. Moon, B.~C. Cha, and S.~K. Kim}, {\em Offsetting inventory cycles
  using mixed integer programming and genetic algorithm}, International Journal
  of Industrial Engineering: Theory, Applications and Practice, 15 (2008),
  pp.~245--256.

\bibitem{MBR}
{\sc N.~N. Murthy, W.~Benton, and P.~A. Rubin}, {\em Offsetting inventory
  cycles of items sharing storage}, European Journal of Operational Research,
  150 (2003), pp.~304--319.

\bibitem{RU}
{\sc R.~A. Russell and T.~L. Urban}, {\em Offsetting inventory replenishment
  cycles}, European Journal of Operational Research, 254 (2016), pp.~105--112.

\bibitem{YC}
{\sc M.~Yao and W.~Chu}, {\em A genetic algorithm for determining optimal
  replenishment cycles to minimize maximum warehouse space requirements},
  Omega, 36 (2008), pp.~619--631.

\bibitem{YCL}
{\sc M.~Yao, W.~Chu, and Y.~Lin}, {\em Determination of replenishment dates for
  restricted-storage, static demand, cyclic replenishment schedule}, Computers
  \& operations research, 35 (2008), pp.~3230--3242.

\end{thebibliography}

\end{document}